\documentclass[aps,pre,reprint,floatfix,letterpaper,nofootinbib]{revtex4-2}
\pdfoutput=1
\usepackage{graphicx}
\usepackage{microtype}
\usepackage{amsmath,amsthm}
\usepackage{amssymb}
\usepackage{subfigure}
\usepackage{hyperref}
\usepackage{url}
\usepackage{mathrsfs}
\usepackage{calrsfs}
\usepackage{amsfonts}
\usepackage{latexsym}
\usepackage{ragged2e}
\usepackage{epsfig}
\usepackage{textcomp}
\usepackage{phaistos}
\usepackage{lipsum}

\makeatletter
\renewcommand\@makefnmark{\hbox{\@textsuperscript{\normalfont\@thefnmark}}}
\renewcommand\@makefntext[1]{%
  \parindent 1em\noindent
            \hb@xt@1.8em{%
                \hss\@textsuperscript{\normalfont\@thefnmark}}#1}
\makeatother

\hypersetup{linktoc=all, hidelinks}
\graphicspath{{Images/}}

\graphicspath{{Images/}}



\def\sideremark#1{\ifvmode\leavevmode\fi\vadjust{\vbox to0pt{\vss
 \hbox to 0pt{\hskip\hsize\hskip1em
 \vbox{\hsize1.5cm\tiny\raggedright\pretolerance10000
 \noindent #1\hfill}\hss}\vbox to8pt{\vfil}\vss}}}%
\begin{document}

\title{A Generalized Landauer's Principle for Unitarily Transformed Thermal Reservoirs}
\author{Hao Xu}
\thanks{Corresponding author}
\email{haoxu@yzu.edu.cn}
\affiliation{Center for Gravitation and Cosmology, College of Physical Science and Technology, Yangzhou University, \\180 Siwangting Road, Yangzhou City, Jiangsu Province 225002, China}

\begin{abstract}
Landauer's principle, a cornerstone of quantum information and thermodynamics, appears to be violated when the thermal reservoir is replaced by a squeezed thermal state (STS), owing to the additional thermodynamic resources inherently present in the squeezed state. We introduce a formal extension of the principle to such unitarily transformed thermal states. By defining an effective Hamiltonian, we rigorously establish a generalized Landauer inequality, which naturally reduces to the standard case for an ordinary thermal reservoir as a special instance. The framework further yields a consistent definition of entropy production and a proof of its non-negativity. We illustrate its utility by studying an arbitrarily moving Unruh-DeWitt detector coupled to a quantum field initially prepared in the STS. Using perturbation theory, we compute the entropy production explicitly, confirming its positivity. As a result of the symmetry breaking induced by the unitary transformation, it depends on both the proper time interval and the absolute spacetime position. Our work clarifies the apparent violation of Landauer's principle with STSs by providing a consistent information-theoretic framework. It also offers a practical tool for computing entropy production in non-equilibrium and relativistic settings, with potential implications for quantum thermal machines and information protocols.

\end{abstract}

\maketitle
The squeezed thermal state (STS) constitutes a fundamental concept in modern quantum thermodynamics and metrology. It represents a non-equilibrium system formed by applying a unitary squeezing operator to a thermal density matrix, resulting in phase-sensitive noise redistribution—fluctuations are suppressed in one quadrature below the standard quantum limit at the expense of amplified noise in the conjugate quadrature \cite{Marian:1993aa,Kim:1989bb}. By extending the notion of coherent thermal states through squeezing transformations, STSs enable enhanced measurement precision and exhibit distinctive thermodynamic properties. Notably, when employed as an engineered thermal reservoir in quantum heat engines, STSs can facilitate performance that surpasses the classical Carnot limit, with potential applications extending to quantum-enhanced thermal management and information processing \cite{Rossnagel:2014cc,Klaers:2017dd}. Research on STSs thus provides critical insights into the interplay between quantum squeezing and thermal fluctuations, offering a productive framework for exploring the interface of quantum optics, statistical mechanics, and open quantum systems \cite{Breuer}.

Given the advantages of STSs as thermal reservoirs, researchers are actively exploring their role in energy and information exchange between quantum systems. A key principle in this domain is Landauer's principle, which originally stated that erasing one bit of information requires a minimum energy cost of $k_BT\ln2$ \cite{landauer1961}. However, in a seminal thought experiment presented in \cite{Klaers:2019,Chattopadhyay:2025vis}, Klaers showed that when a STS is employed as the thermal bath for information erasure, the conventional form of Landauer's principle can be violated, and the associated work cost falls below the standard Landauer limit. This is because the STS inherently contains additional thermodynamic resources, rather than violating thermodynamic laws. This work highlights the potential of squeezed thermal environments for pioneering low-energy quantum information processing and advancing the design of energy-efficient electronics.

On the other hand, the apparent violation of Landauer's principle when using a STS for information erasure is consistent with its mathematical formulation, which is derived from a system-reservoir interaction model \cite{Reeb2013}. The derivation rests on four key assumptions: (i) both the system $S$ and the reservoir $R$ are described by Hilbert spaces, (ii) $R$ is initially in a thermal state with inverse temperature $\beta$, (iii) $S$ and $R$ are initially uncorrelated, and (iv) the process proceeds via unitary evolution. If all four assumptions are satisfied, then Landauer's principle can be expressed as 
\begin{equation}
\beta \Delta Q \geqslant \Delta S.
\label{bound}
\end{equation}
Here, $\Delta Q := \text{tr}\left[\hat{H}_R(\rho'_R - \rho_R) \right]$ is the heat transferred to the reservoir $R$, where $\hat{H}_R$ is the Hamiltonian of $R$, while $\rho'_R$ and $\rho_R$ denote the final and initial states of $R$, respectively. Moreover, $\Delta S := S(\rho_S) - S(\rho'_S)$ is the von Neumann entropy change between the initial state $\rho_S$ and the final state $\rho'_S$ of the system $S$.

The violation observed in the squeezed thermal memory scenario arises precisely because assumption (ii) is not met: a STS is a nonequilibrium reservoir, not a standard Gibbs state. Therefore, the thought experimental violation does not contradict the rigorously derived principle but rather confirms its sensitivity to the nature of the thermal environment.

Landauer's principle derives its significance from its universality: provided that the four assumptions are satisfied, the inequality \eqref{bound} holds universally, regardless of the specific system or interaction type. Additionally, as the von Neumann entropy remains invariant under unitary evolution, it cannot serve as an indicator of whether a system's evolution is reversible. In contrast, Landauer's principle provides a rigorous framework precisely for defining and calculating entropy production \cite{Batalhao2015,Landi:2020bsq}, which is the key quantity that measures the irreversibility of quantum interactions. The discrepancy between the two sides of equation \eqref{bound} exactly equals the entropy production, which can be decomposed into the mutual information between $S$ and $R$ and the relative entropy between the initial and final states of $R$. This observation is also directly related to the fluctuation theorem \cite{Esposito2009}. Landauer's principle converts quantum correlations in a bipartite system into computable physical quantities associated with each subsystem. This decomposition emphasises how the principle bridges information theory and thermodynamics, providing a clear metric for energy dissipation during processes like information erasure.

However, once Landauer's principle is violated, this universality is lost. In such scenarios, the calculation and even the definition of entropy production become considerably challenging. This difficulty arises from the fact that different quantum systems may possess vastly different Hamiltonians and energy spectra, making the analysis of their joint Hilbert space after interaction highly complex. Evaluation becomes inherently complex due to the need to account for non-equilibrium dynamics, reservoir engineering, and emergent correlations. Such computations are only tractable in certain special cases, such as within the framework of Gaussian quantum mechanics \cite{Adesso2014,Ptaszynski2019,Xu:2021ihm}. This underscores the principle's role as a benchmark: deviations from it reveal fundamental limits of thermodynamic laws in quantum regimes.

The four assumptions introduced in \cite{Reeb2013} constrain the form of the system–environment interaction. It is widely accepted that these assumptions constitute a general and minimal set of conditions. A natural direction for subsequent theoretical work, therefore, is to examine the specific physical role of these assumptions and explore whether modifying or extending them can lead to richer physical insights. Research in this direction has generally followed two main approaches. The first one maintains the original four assumptions, and proceeds by either imposing stricter constraints (e.g., introducing a fifth assumption) to derive tighter bounds, or by examining the behavior of specific physical quantities within particular models\cite{Goold2015,Lorenzo:2015fgu,Strasberg:2017klo,Guarnieri:2017tan,Campbell:2017pjc,Timpanaro2020,Vu2022,Rolandi:2022ose}. The second approach relaxes the original assumptions to derive perturbative corrections to the inequality \cite{Iyoda2017,Xu:2024xlx}. Our work diverges from both approaches. This work provides a natural and analytic extension of the original framework and demonstrates Landauer's principle can be formally extended to the case where the initial thermal reservoir state is replaced by a unitarily transformed thermal state $\hat O \rho_{\mathrm{th}} \hat O^\dagger$, requiring neither complex functions nor perturbative approximations. Here $\hat O$ is a unitary operator and $\rho_{\mathrm{th}}$ denotes the thermal state---a class that includes, for example, the STS. The modified principle takes the form of an inequality, obtained by replacing the original reservoir Hamiltonian $H_R$ with an effective Hamiltonian $\hat{H}_{\mathrm{eff}} = \hat O \hat{H}_R \hat O^\dagger$. For a system $S$ interacting with such a generalized reservoir, the entropy production can be identified directly as the difference between the two sides of the revised inequality, which greatly simplifies its evaluation. 

It is useful to clarify how the present approach relates to existing strategies for treating non-thermal reservoirs. When the reservoir is not in a Gibbs state, two broad strategies exist. The first is resource-theoretic: one quantifies the non-equilibrium resource present in the initial state, such as the ergotropy of a squeezed thermal state \cite{Nieuwenhuizen:2004exd}, and incorporates it as an additional thermodynamic resource \cite{Brandao:2013cgt}, thereby preserving the physical Hamiltonian in all definitions. This approach has the advantage of maintaining direct contact with physical energy flows, but it requires characterizing the resource separately for each class of reservoir states. The second strategy, which we adopt here, is to absorb the non-equilibrium preparation into an effective Hamiltonian $\hat{H}_{\text{eff}}$. This restores the formal structure of the standard Landauer's principle at the cost that the left-hand side involves $\Delta \hat{H}_{\text{eff}}$ rather than the physical heat. Our result is therefore best understood as an information-theoretic identity involving the modular Hamiltonian of the initial reservoir state, rather than as a direct thermodynamic bound on heat dissipation. Its practical virtue is that it provides a closed, analytic formula for the entropy production that holds without perturbative approximations, without restriction to Gaussian states, and without requiring a separate quantification of ergotropy.

Moreover, we illustrate this framework with a spin-boson model, analyzing the entropy production for an arbitrarily moving Unruh--DeWitt detector coupled to a quantum field theory (QFT) prepared in a STS. The Unruh-DeWitt detector was originally introduced to study the Unruh effect, whereby a uniformly accelerated observer in vacuum detects thermal radiation proportional to its acceleration \cite{Unruh:1976db,Unruh:1983ms,DeWitt1979,Ben-Benjamin:2019opz}. Moreover, the Unruh-DeWitt detector model provides a simplified yet powerful framework for studying light-matter interactions in relativistic settings. Through the equivalence principle, the Unruh-DeWitt detector also serves as a vital tool for studying Hawking radiation, enabling precise and observable analyses of different spacetime backgrounds, detector trajectories, and energy-transfer mechanisms. Consequently, it has become one of the most versatile models for investigating processes of energy and information transfer, offering potential insights into open questions such as the black hole information loss paradox \cite{Mathur:2009hf}. We compute the time evolution of the system, and the entropy production is evaluated with the formula provided in Theorem 1 below, and its positivity is explicitly verified. The analytical approach and conclusions presented here can be naturally extended to other models.

\newtheorem{theorem}{Theorem}
\begin{theorem}
Let $\rho_{SR} = \rho_S \otimes \rho_R$ be a product state on a bipartite system $SR$, 
where $\rho_R = \hat O \frac{\mathrm e^{-\beta \hat{H}_R}}{\operatorname{tr}[\mathrm e^{-\beta \hat{H}_R}]} \hat O^\dagger$ 
is obtained from the thermal state at inverse temperature $\beta$ by applying a unitary operator $\hat O$. 
Let $\hat{U}$ be a global unitary, and denote the evolved state by $\rho'_{SR} = \hat{U}(\rho_S \otimes \rho_R)\hat{U}^\dagger$, 
with reduced states $\rho'_S$ and $\rho'_R$. Then
\begin{align*}
\beta \Delta \hat{H}_{\text{eff}}&:=\beta\operatorname{tr}\bigl[\hat{H}_{\text{eff}}(\rho'_R-\rho_R)\bigr] \nonumber \\
 &= \Delta S + I(S':R') + D(\rho'_R \|\rho_R),
\end{align*}
where $\hat{H}_{\mathrm{eff}}:= \hat O \hat{H}_{R} \hat O^{\dagger}$, $\Delta S:=S(\rho_S)-S(\rho'_S)$, and $I(S':R')$ and $D(\rho'_R \|\rho_R)$ denote the mutual information and relative entropy, respectively. 
Consequently, Landauer's inequality holds:
\begin{align*}
\beta \Delta \hat{H}_{\text{eff}} \geqslant \Delta S .
\end{align*}
\end{theorem}
To obtain the theorem, we need to utilize the unitarity of time evolution (the total von Neumann entropy of the system remains constant), the product form of the initial state (the initial von Neumann entropy equals the sum of the entropies of $S$ and $R$), and
$
\hat O\, \mathrm{e}^{-\beta H_R}\, \hat O^\dagger = \mathrm{e}^{-\beta \hat O H_R \hat O^\dagger} = \mathrm{e}^{-\beta H_{\mathrm{eff}}}.
$ See Appendix~\ref{app:proof} for a detailed proof.

Theorem 1, in its mathematical structure, is the Landauer identity expressed in terms of the rotated Hamiltonian $\hat{H}_{\text{eff}}$. The contribution of this work is not just the identity itself, but the recognition that it can be applied systematically to any reservoir prepared in a unitarily transformed thermal state, thereby yielding an exact, non-perturbative formula for the entropy production, $\beta \Delta \hat{H}_{\text{eff}}- \Delta S = I(S':R') + D(\rho'_R \|\rho_R)$. This formula retains the simplicity of the original Landauer relation while extending its domain of applicability. The standard Landauer's principle, with the physical Hamiltonian $H_R$ and the physical heat $\Delta Q$, is recovered as the special case $\hat{O} = I$.

We now apply Theorem 1 to analyze the entropy production in the interaction between an Unruh-DeWitt detector and a QFT prepared initially in a STS. The free Hamiltonians of the detector and the field are respectively
$\hat{H}^{(d)}_0 = \frac{\Omega}{2}\,\sigma_z ,
\hat{H}^{(f)}_0 = \sum_{j=1}^{\infty} \omega_j \hat{a}_j^\dagger \hat{a}_j$ \cite{footnote}. The interaction Hamiltonian
$
\hat{H}_{\text{int}} = \lambda\,\hat{\sigma}_x(\tau)\,\hat{\phi}[x(\tau)] 
\;\propto\; (\hat{\sigma}_{-} + \hat{\sigma}_{+})(\hat{a} + \hat{a}^{\dagger})
$,
where \(\lambda\) is the coupling strength, $\hat{\sigma}_{+}$ and $\hat{\sigma}_{-}$ are the detector’s raising and lowering operators, and \(\hat{a}, \hat{a}^{\dagger}\) denote the field's annihilation and creation operators. The $\hat{\sigma}_x$ coupling in the interaction Hamiltonian enables energy exchange of the detector and the field. In contrast, the model with $\hat{\sigma}_z$ coupling, while exactly solvable, only induces decoherence in the detector's density matrix without enabling net heat transfer. It therefore cannot describe thermalization or genuine energy-flow processes in the detector. For this reason, in this work we analyze the $\hat{\sigma}_x$ model using a perturbative approach.

The interaction Hamiltonian corresponds to two kinds of physical process:
\begin{itemize}
    \item The \textit{rotating wave terms} $\hat{\sigma}_{+}\hat{a} + \hat{\sigma}_{-} \hat{a}^{\dagger}$ describe energy-conserving processes such as stimulated emission and absorption. In the rotating-wave approximation, it reduces to the Jaynes--Cummings model~\cite{Jaynes1963}.
    \item The \textit{counter-rotating wave terms} $\hat{\sigma}_{+} \hat{a}^{\dagger} + \hat{\sigma}_{-}\hat{a}$ account for non-energy-conserving processes, which are essential for manifestations of the Unruh effect and related phenomena.
\end{itemize}

In the interaction picture, the detector operator evolves as $\hat{\sigma}_x(\tau) = \hat{\sigma}_{+} e^{i\Omega \tau} + \hat{\sigma}_{-} e^{-i\Omega \tau}$, while the field operator takes the form
\begin{align}
\hat{\phi}[x(\tau)] = \sum_{j=1}^{\infty} \left( \hat{a}_j e^{-i\omega_j t(\tau)} u_j[x(\tau)] + \hat{a}^{\dagger}_j e^{i\omega_j t(\tau)} u_j^*[x(\tau)] \right),
\label{int}
\end{align}
where the mode $u_j[x(\tau)]$ depends on the boundary conditions. From $\tau=0$ to $\tau=T$, the time evolution operator is given by the Dyson series:
\begin{align*}
\hat{U}(T,0) = &\; 1
\underbrace{ - i \int^{T}_{0} d\tau \, \hat{H}_{\text{int}}(\tau) }_{\hat{U}^{(1)}} \nonumber\\
&\underbrace{ + (-i)^2 \int^{T}_{0} d\tau \int^{\tau}_{0} d\tau' \, \hat{H}_{\text{int}}(\tau) \hat{H}_{\text{int}}(\tau') }_{\hat{U}^{(2)}} + \dots 
\label{dyson}
\end{align*}
The density matrix at time $\tau = T$ can be expanded order by order as
\begin{equation}
\rho_{T} = \rho^{(0)}_{T} + \rho^{(1)}_{T} + \rho^{(2)}_{T} + \mathcal{O}(\lambda^3),
\end{equation}
with
\begin{align}
\rho^{(0)}_{T} &= \rho_0, \nonumber\\
\rho^{(1)}_{T} &= \hat{U}^{(1)} \rho_0 + \rho_0 \hat{U}^{(1)\dagger},\nonumber \\
\rho^{(2)}_{T} &= \hat{U}^{(1)} \rho_0 \hat{U}^{(1)\dagger} + \hat{U}^{(2)} \rho_0 + \rho_0 \hat{U}^{(2)\dagger}.
\label{rho}
\end{align}

We choose the initial state of the total system $\rho_{0}= \rho^{\text{(0)}}_{d}  \otimes \rho^{\text{(0)}}_{f}$, where the state of the qubit is
\begin{equation}
\rho^{\text{(0)}}_d=\begin{pmatrix}
    p & 0 \\

    0 & 1-p
\end{pmatrix},
\label{dq}
\end{equation}
and $0< p< 1$ is a real number, while the state of the QFT is
\begin{equation}
\rho^{\text{(0)}}_{f}= \bigotimes_{j=1}^{\infty}\hat{S}_j \rho^{\text{th}}_{j}\hat{S}_j^{\dagger}.
\end{equation}
The $\rho^{\text{th}}_{j}$ and $\hat{S}_j$ are the density matrix of thermal state and squeezing operator in the modes $k_j$ respectively, satisfying 
\begin{align}
\text{Tr}\left(\hat{a}_j^{\dagger}\hat{a}_j\rho^{\text{th}}_{j}\right)=\frac{1}{e^{\beta \omega_j}-1}:=\bar{n}_j,
\end{align}
and
\begin{align}
\hat{S}_j=\exp\left[ \frac{1}{2} \xi_{j}^{*} \hat{a}_{j}^{2} - \frac{1}{2} \xi_{j} (\hat{a}_{j}^{\dagger})^{2} \right], \quad \xi_{j} = r_j e^{i\theta_j} \; (r_j \geqslant 0),
\end{align}
where the $r_j$ and $\theta_j$ indicate the squeezing strength and phase, respectively. Defining $\langle \hat{A}\rangle:=\text{Tr}\left[\hat{A}\rho^{\text{(0)}}_{f}\right]$ and using 
\begin{align}
\hat{S}_{j}^{\dagger}(r_{j})\hat{a}_{j}\hat{S}_{j}(r_{j})&=\hat{a}_{j}\cosh r_{j}-\hat{a}_{j}^{\dagger}e^{i\theta_{j}}\sinh r_{j},\nonumber \\
\hat{S}_{j}^{\dagger}(r_{j})\hat{a}_{j}^{\dagger}\hat{S}_{j}(r_{j})&=\hat{a}_{j}^{\dagger}\cosh r_{j}-\hat{a}_{j}e^{-i\theta_{j}}\sinh r_{j},
\end{align}
we have $\langle \hat{a}_{j}\rangle =\langle \hat{a}_{j}^{\dagger}\rangle=0$, while
\begin{align}
    \langle \left( \hat{a}_{j} \right)^{2}\rangle &= \langle ( \hat{a}_{j}^{\dagger} )^{2} \rangle^{*} = -\frac{1}{2} \left( 2\bar{n}_{j} + 1 \right) \sinh\left( 2r_{j} \right) e^{i\theta_{j}},\nonumber \\
    \langle \hat{a}_{j}^{\dagger} \hat{a}_{j} \rangle &=\langle \hat{a}_{j} \hat{a}_{j}^{\dagger} \rangle - 1= \bar{n}_{j} \cosh\left( 2r_{j} \right) + \sinh^{2} r_{j}.
\end{align}
In the order of $\lambda$, by directly applying the $\hat{U}^{(1)}$ to the initial density matrix $\rho_{0}$ and then performing the partial trace over the detector and the QFT, we will find that $\text{Tr}\left(\sigma_{\pm}\rho^{\text{(0)}}_d\right)=0$ and $\langle \phi(\tau)\rangle=0$, which implies that the reduced density matrix of the QFT and detector are both zero.

For the order of $\lambda^2$, defining $\Delta \tau=\tau_1-\tau_2$, we have
\begin{widetext}
\begin{align}
\hat{U}^{(1)}\rho_0 \hat{U}^{(1)\dagger}&=\lambda^2\iint d\tau_1 d\tau_2 \begin{pmatrix}
    (1-p)e^{-i\Omega\Delta \tau} & 0 \\
    0 & pe^{i\Omega\Delta \tau}
\end{pmatrix} \hat{\phi}(\tau_2)\rho^{\text{(0)}}_{f}\hat{\phi}(\tau_1), \nonumber\\
\hat{U}^{(2)}\rho_0 & =- \theta(\Delta \tau) \lambda^2 \iint d\tau_1 d\tau_2 \begin{pmatrix}
    pe^{i\Omega\Delta \tau} & 0 \\
    0 & (1-p)e^{-i\Omega\Delta \tau}
\end{pmatrix} \hat{\phi}(\tau_1) \hat{\phi}(\tau_2) \rho^{\text{(0)}}_{f},\nonumber \\
\rho_0 \hat{U}^{(2)\dagger} &=- \theta(-\Delta \tau) \lambda^2 \iint d\tau_1 d\tau_2 \begin{pmatrix}
    pe^{i\Omega\Delta \tau} & 0 \\
    0 & (1-p)e^{-i\Omega\Delta \tau}
\end{pmatrix} \rho^{\text{(0)}}_{f} \hat{\phi}(\tau_1) \hat{\phi}(\tau_2).
\end{align}
\end{widetext}
By taking the partial trace over the QFT sector in the three expressions above and utilizing the identity $\theta(\Delta \tau) + \theta(-\Delta \tau) = 1$ for the Heaviside step function, we can derive the reduced density matrix for the detector subsystem, expressed as
$
\begin{pmatrix}
    \delta p & 0 \\

    0 & -\delta p
\end{pmatrix},
$
where
\begin{equation}
\begin{aligned}
\delta p = \lambda^2 \iint d \tau_1 d\tau_2 \left( (1-p)e^{-i\Omega\Delta \tau}-pe^{i\Omega\Delta \tau}\right)  \langle\hat{\phi}(\tau_1) \hat{\phi}(\tau_2)\rangle .
\end{aligned}
\label{p}
\end{equation}
The general form of the two-point function of STS can be expressed as
\begin{align}
&\langle\hat{\phi}(\tau_1) \hat{\phi}(\tau_2)\rangle = \sum_{j=1}\big[ \nonumber \\
&\quad \langle \hat{a}_j \hat{a}^{\dagger}_j \rangle e^{-i\omega_j(t(\tau_1)-t(\tau_2))} u_j[x(\tau_1)]u^*_j[x(\tau_2)] \nonumber \\
&\quad + \langle \hat{a}^{\dagger}_j \hat{a}_j \rangle e^{i\omega_j(t(\tau_1)-t(\tau_2))} u^*_j[x(\tau_1)]u_j[x(\tau_2)] \nonumber \\
&\quad + 2\operatorname{Re} \Big(\langle (\hat{a}_j^\dagger)^2 \rangle e^{i\omega_j(t(\tau_1)+t(\tau_2))} \nonumber \\
&\qquad \times u^*_j[x(\tau_1)]u^*_j[x(\tau_2)]\Big) \big].
\label{two}
\end{align}
By substituting the above formula into the $\delta p$, we obtain its specific expression
\begin{align}
\delta p= \lambda^2 \sum_{j}&\bigg[ \left( (1-p) \langle \hat{a}_j \hat{a}_j^{\dagger} \rangle - p \langle \hat{a}_j^\dagger \hat{a}_j \rangle \right)|I_{+,j}|^2 \nonumber \\
&+ \left( (1-p) \langle \hat{a}_j^\dagger \hat{a}_j \rangle - p \langle \hat{a}_j \hat{a}_j^\dagger \rangle \right) |I_{-,j}|^2 \nonumber\\ 
&+ 2(1-2p) \operatorname{Re} \left( \langle (\hat{a}_j^\dagger)^2 \rangle I_{+,j} I_{-,j} \right)\bigg],
\label{p2}
\end{align}
where
\begin{equation}
I_{\pm,j}:=\int^{s}_0 d\tau~e^{i\left[\pm \Omega \tau+\omega_jt(\tau)\right]}u_j^*[x(\tau)].
\label{I}
\end{equation}
The $I_{-,j}^{(*)}$ and $I_{+,j}^{(*)}$ correspond to the rotating and counter-rotating wave terms, respectively. For finite interaction times, $I_{+,j}$ is generally nonzero even for inertial motion because of switching transients. Resonant rotating-wave response and vanishing vacuum excitation for inertial motion refer to the appropriate stationary, long-time limit. The familiar Unruh thermal response is a special case for uniform acceleration in the Minkowski vacuum under suitable stationary or long-time conditions. These conditions and the switching prescription are discussed in the conclusion. Moreover, $I_{\pm,j}$ incorporates the dispersion relation of the QFT, allowing it to compute the response of a detector undergoing arbitrary motion and energy spectra, even in the case of Lorentz-violating theories \cite{Hammad:2021rey,Xu:2025smb}. In fact, a current strategy for detecting the Unruh effect involves adjusting the particle's trajectory and the properties of the thermal reservoir to amplify the terms relating $I_{+,j}^{(*)}$, thereby making its signature more pronounced \cite{Kempf2022}.

The entropy change in the Landauer's principle
\begin{equation}
\Delta S = -\ln{\frac{1-p}{p}}\delta p.
\end{equation}
We now provide a brief discussion of the two-point function. The expectation value of terms linear in creation and annihilation operators vanishes for a STS, in contrast to other states such as a coherent state \cite{Simidzija:2017jpo,Simidzija:2017kty}. When considering the contribution of second-order terms to the detector, the quantity $\delta p$ contains terms proportional to $|I_{-,j}|^2$ and $|I_{+,j}|^2$
, which, as noted earlier, correspond to the rotating wave and counter-rotating wave contributions, respectively. Notably, these terms depend only on the proper interval along the detector’s trajectory. \emph{In contrast, the cross term $I_{+,j} I_{-,j}$ depends not only on the proper time separation but also on the specific spacetime coordinates of the trajectory.} 

In the case of the vacuum state $|0\rangle$, which possesses maximal symmetry, the vacuum expectation value of the two-point function must be a Lorentz-invariant function—namely, a function of the proper interval between the two points. For non-vacuum states, however, the breaking of symmetries and the resulting complexity in functional dependence mean that this simplification no longer holds. The quantum field background may not be homogeneous in spacetime, nor is the state necessarily an eigenstate of the Hamiltonian. Consequently, the two-point function can depend explicitly on the absolute positions of the points, indicating that correlations in the system are sensitive to the absolute clock time of the experiment. Such states possess an intrinsic ``starting time'' or coherence, and their time evolution breaks temporal translation symmetry of the correlation functions. The presence of off-diagonal elements in the density matrix of these states defines preferred spacetime points or directions, thereby explicitly breaking the overall translational symmetry. The unitary operator $\hat{O}$ introduced in Theorem 1---for instance, a squeezing operator---can potentially break the symmetry of the QFT. This effect manifests explicitly in the cross terms between $I_{+,j}^{(*)}$ and $I_{-,j}^{(*)}$.

For the expectation value of the $\hat{H}_{\text{eff}}$ in the QFT, we have
\begin{align}
\hat{H}_{\text{eff}} &= \sum_j^{\infty} \hat{S}_j \, \omega_j \, \hat{a}_j^\dagger \hat{a}_j \, \hat{S}_j^\dagger \nonumber \\ 
&= \sum_j^{\infty} \omega_j \bigg( \cosh (2r_j) \, \hat{a}_j^\dagger \hat{a}_j 
+ \frac{1}{2} \sinh (2r_j) \big[ e^{i\theta_j} (\hat{a}_j^\dagger)^2 \nonumber \\ 
&+ e^{-i\theta_j} (\hat{a}_j)^2 \big] 
+ \sinh^2 r_j \bigg).
\end{align}

It is worth emphasizing that $\hat{H}_{\text{eff}}$ is a mathematical construct. It is the modular Hamiltonian associated with the initial reservoir state $\rho_R$, and its eigenvalues do not correspond to observable physical energy. Consequently, $\Delta\hat{H}_{\text{eff}}$ is not the physical heat exchanged with the reservoir; that quantity is $\Delta\hat{H}_R$, which can differ substantially from $\Delta \hat{H}_{\text{eff}}$ when squeezing is present. The apparent violation of Landauer's principle reported in \cite{Klaers:2019} refers precisely to the fact that $\Delta\hat{H}_R$ can fall below $\Delta S/\beta$ for a squeezed reservoir. Theorem 1 does not contradict this observation. Rather, it shows that the standard form of the inequality is recovered if one replaces $\hat{H}_R$ by $\hat{H}_{\text{eff}}$. Whether such a replacement is physically meaningful depends on the context. For information-processing purposes, where the quantity of interest is the entropy balance rather than the physical heat, our identity provides a useful decomposition. For thermodynamic analyses concerned with actual heat flows, the conventional Landauer bound remains the appropriate reference, and the apparent violation with respect to $\hat{H}_R$ signals the presence of extractable work (ergotropy) in the initial squeezed state \cite{Nieuwenhuizen:2004exd}.

Evaluating the expectation value of $\hat{H}_{\text{eff}}$ with the reduced density matrix of the QFT at $\lambda^2$ order and then subtracting $\Delta S$, we finally obtain (see Appendix~\ref{app:entropy})
\begin{widetext}
\begin{align}
    \beta\Delta \hat{H}_{\text{eff}}-\Delta S
    &=\lambda^2\sum_j\left[A_j|I_{-,j}|^2+B_j|I_{+,j}|^2-C_j\operatorname{Re}\left(e^{-i\theta_j}I_{+,j}I_{-,j}\right)\right] \nonumber\\
    &=\lambda^2\sum_j\left[A_j\left|I_{-,j}-\frac{C_j}{2A_j}I_{+,j}^*e^{i\theta_j}\right|^2+\left(B_j-\frac{C_j^2}{4A_j}\right)|I_{+,j}|^2\right].
\label{ep}
\end{align}
\end{widetext}
The coefficients are respectively
\begin{align}
    A_{j} &= \sinh^{2} r_{j} \left[ A_{j}^{\text{(min)}}+B_{j}^{\text{(min)}} \right] + A_{j}^{\text{(min)}}, \nonumber\\
    B_{j} &= \sinh^{2} r_{j} \left[ A_{j}^{\text{(min)}}+B_{j}^{\text{(min)}} \right] + B_{j}^{\text{(min)}}, \nonumber\\
    C_{j} &= 2 \sqrt{ \sinh^{2} r_{j}(1 + \sinh^{2} r_{j}) } \left[ A_{j}^{\text{(min)}}+B_{j}^{\text{(min)}} \right],
\end{align}
where
\begin{align}
    A_{j}^{\text{(min)}} &= [(\bar{n}_j+1)p-\bar{n}_j(1-p)] \left[ \beta \omega_{j} - \ln \frac{1-p}{p} \right]\geq 0, \nonumber\\
    B_{j}^{\text{(min)}} &= [(\bar{n}_j+1)(1-p)-\bar{n}_jp] \left[ \beta \omega_{j} + \ln \frac{1-p}{p} \right]\geq 0,
\end{align}
correspond to the entropy production for the canonical thermal state \cite{Xu:2021buk,Xu:2024ztq}. We have 
\begin{align}
   4 A_{j} B_{j} -C_{j} ^2= 4A_{j}^{\text{(min)}}B_{j}^{\text{(min)}}\geq 0.
\end{align}
Together with $A_j,B_j\geq0$, this relation establishes non-negativity of the quadratic form in the first line of Eq.~\eqref{ep}. The completed-square expression in the second line applies to modes with $A_j>0$. If $A_j=0$, the determinant condition forces $C_j=0$, leaving the non-negative contribution $\lambda^2B_j|I_{+,j}|^2$. In particular, $A_j=0$ occurs at $r_j=0$ and $p=\bar n_j/(2\bar n_j+1)$; this is a population condition and does not require vanishing interaction.

Zhang \emph{et al.}~\cite{Zhang2023NER} also study reservoirs prepared from thermal states by unitary driving. Their Landauer-like relation retains the physical reservoir Hamiltonian and explicitly includes the change in the nonequilibrium resource, measured relative to a reference thermal state. Both approaches use the same mutual-information and relative-entropy balance. Here, the initial transformed state itself serves as the reference, and $\hat H_{\mathrm{eff}}$ yields a manifestly non-negative entropy balance rather than a bound on physical heat alone. The formulations therefore complement one another. Earlier detector--field studies~\cite{Xu:2021buk,Xu:2024ztq} established the thermal-reservoir case and analyzed rotating and counter-rotating contributions for general motion. The present application extends that analysis to squeezed thermal fields, including phase-dependent anomalous correlations and an explicit non-negative quadratic form for entropy production.

In conclusion, we have presented a generalized Landauer's principle that extends the standard relative-entropy formulation to reservoirs prepared in unitarily transformed thermal states. By introducing the effective (modular) Hamiltonian $\hat{H}_{\text{eff}} = \hat{O} \hat{H}_R \hat{O}^\dagger$, we have shown that the formal structure of the Landauer relation, including the decomposition of entropy production into mutual information and relative entropy, carries over without modification. The result is an information-theoretic identity with a clear practical benefit: it yields a closed, analytic expression for entropy production for any reservoir in the considered class, without requiring perturbative approximations, restrictions to Gaussian states, or prior characterization of ergotropy. We have clarified that $\Delta \hat{H}_{\text{eff}}$ is not the physical heat, and that the apparent violation of the standard Landauer bound with respect to $\hat{H}_R$ reflects the non-equilibrium resources present in the initial reservoir state. The framework thus provides a consistent and computationally convenient tool for analyzing entropy balance in systems coupled to engineered reservoirs.

The application of this framework to the Unruh-DeWitt detector model coupled to a STS vividly illustrates its utility. The perturbative calculation confirms that entropy production is indeed non-negative, even for detectors on arbitrary trajectories in a non-equilibrium field state. Our choice of the model is motivated precisely \emph{by its universality, not its specificity}. While in our example the field is in a squeezed thermal state, the two-point function \eqref{two} is in fact a general result (when the one-point function vanishes). This structure necessarily forces the entropy change \( \Delta S \) to decompose into a sum of modulus-square terms from \( I_{+,j} \), \( I_{-,j} \), and their cross-terms. Theorem 1 then requires \( \Delta H_{\text{eff}} \) to contain matching terms to satisfy the inequality. Consequently, our final result for the entropy production \eqref{ep} is a general consequence of this structure. It will hold for other unitary operators \( \hat{O} \), with only the specific forms of \( A_j, B_j, C_j \) changing. Thus, \emph{our aim is not to present a special case, but to argue for a general pattern arising from symmetry considerations in the interaction.}

Given the model's generality, one can naturally examine its physical implications. A possible route toward enhancing counter-rotating excitation is to combine control of the detector trajectory with preparation of the field state. From (14), the detector's transition probability receives contributions from both $|I_{+,j}|^2$ and $|I_{-,j}|^2$, weighted by the statistical factors $\langle\hat{a}_j \hat{a}_j^\dagger\rangle$ and $\langle\hat{a}_j^\dagger \hat{a}_j\rangle$, respectively. To isolate the counter-rotating contribution, one should suppress $|I_{-,j}|^2$ while enhancing $|I_{+,j}|^2$. The magnitude of $|I_{\pm,j}|^2$ is largely determined by the detector's trajectory. As shown in \cite{Kempf2022}, a suitably engineered trajectory can make $I_{-,j}$ vanish while keeping $I_{+,j}$ finite. In that limit, and with the detector initially in its ground state ($p \to 0$), the transition probability reduces to
\begin{align}
\delta p\big|_{p\to 0, I_{-,j}=0} = \lambda^2 \sum_j \big(\bar{n}_j \cosh 2r_j + \sinh^2 r_j + 1\big) |I_{+,j}|^2.
\end{align}
In Eq.~(22), $\bar n_j$ is the mean particle number of the unsqueezed thermal mode, while $r_j$ specifies its squeezing strength. These parameters determine the field-state prefactor, whereas the detector trajectory determines $|I_{+,j}|^2$ and $|I_{-,j}|^2$. Equation~(22) therefore describes a trajectory-selected counter-rotating contribution whose magnitude can be increased through joint control of the trajectory, thermal occupation, and squeezing. In Eq.~\eqref{p2}, the nonzero occupation of the squeezed field also weights rotating-wave absorption, which alone is not an acceleration-specific signature. This response differs from the conventional Unruh setting with constant proper acceleration. In that setting, both $I_{+,j}$ and $I_{-,j}$ are generally nonzero, so one must return to Eq.~\eqref{p2}, retaining both contributions and any state-dependent interference. Increasing excitation through field-state preparation alone should therefore be distinguished from the trajectory-dependent response associated with acceleration. The practical relevance of this enhancement can be made concrete. Consider a circuit-QED architecture in which a superconducting qubit serves as the Unruh-DeWitt detector, with its effective acceleration simulated by modulating the qubit-field coupling in time, and a squeezed thermal microwave field is prepared using a Josephson parametric amplifier \cite{Hao:2025jao}. In such a setup, squeezing increases the statistical weight of the selected counter-rotating excitation channel. While a detailed feasibility analysis lies beyond the scope of the present work, \eqref{p2} makes the enhancement mechanism transparent within our framework, and \eqref{ep} guarantees the thermodynamic consistency of any such protocol.

The detector calculation gives a mode-resolved entropy balance for an arbitrary prescribed trajectory, without evaluating a specific acceleration or switching protocol. For definiteness, the mode sums can be regularized in a finite periodic box in $(3+1)$-dimensional Minkowski spacetime with an ultraviolet cutoff. The continuum limit is obtained by taking the box size to infinity, so that the mode spacing tends to zero and the discrete sum becomes a momentum integral with the corresponding density of states. This infinite-volume limit can be taken at fixed ultraviolet cutoff. Equation~\eqref{I} corresponds to a rectangular switching function on $[0,s]$. A smooth real switching function instead gives $\hat H_{\mathrm{int}}(\tau)=\lambda\chi(\tau)\hat\sigma_x(\tau)\hat\phi[x(\tau)]$, with $I_{\pm,j}[\chi]$ obtained by inserting $\chi(\tau)$ into Eq.~\eqref{I}. Such switching changes the response but preserves the modewise quadratic structure and the entropy identity. Finite-time inertial excitation and the dependence on switching and regularization are well established~\cite{Sriramkumar1996,Satz2007}. A thermal detailed-balance relation at $T_U=a/(2\pi)$ applies to the full stationary response of a uniformly accelerated detector in the Minkowski vacuum, with suitable long-time switching limits~\cite{Fewster2016}. It is not a universal relation for each finite-time mode integral, arbitrary trajectories, or squeezed thermal fields. From the interaction-Hamiltonian perspective, the $\hat\sigma_+\hat a^\dagger$ term describes simultaneous excitation of the detector and field and provides the counter-rotating mechanism underlying Unruh excitation~\cite{Kempf2022}. Our focus is on this channel for general detector trajectories, while the conventional thermal Unruh response for uniform rectilinear acceleration in the Minkowski vacuum is a particular case.

Beyond this specific example, this work bridges quantum information theory, thermodynamics, and quantum field theory, offering a precise tool to quantify irreversibility in complex scenarios. It paves the way for future investigations into quantum thermal machines fueled by engineered reservoirs, energy transfer in curved spacetime, and the fundamental limits of information processing in relativistic quantum systems.

\begin{acknowledgments}
Hao Xu thanks National Natural Science Foundation of China (No.12205250) for funding support. He also thanks the anonymous reviewers for their constructive comments on this work.
\end{acknowledgments}

\appendix
\begin{widetext}
\allowdisplaybreaks[1]
\section{Proof of Theorem 1}\label{app:proof}
\begin{proof}
Using $\rho_{SR}=\rho_S\otimes\rho_R$, the unitarity of $\hat{U}$, and the identity
$\hat O\mathrm e^{-\beta H_R}\hat O^\dagger=\mathrm e^{-\beta \hat O H_R \hat O^\dagger}=\mathrm e^{-\beta H_{\mathrm{eff}}}$, we have
\begin{align*}
\Delta S+I(S':R')
&=S(\rho_S)-S(\rho'_S)+I(S':R')\\
&=S(\rho'_R)-S(\rho_R) \\
&=-\operatorname{tr}\bigl[\rho'_R\ln\rho'_R\bigr]
   +\operatorname{tr}\!\Bigl[\rho_R\ln\rho_R\Bigr] \\
&=-\operatorname{tr}\bigl[\rho'_R\ln\rho'_R\bigr]
   +\operatorname{tr}\!\Bigl[\rho_R\ln\!\Bigl(\hat O
     \frac{\mathrm e^{-\beta \hat{H}_R}}{\operatorname{tr}[\mathrm e^{-\beta \hat{H}_R}]}\hat O^\dagger\Bigr)\Bigr] \\
&=-\operatorname{tr}\bigl[\rho'_R\ln\rho'_R\bigr]
   +\operatorname{tr}\!\Bigl[\rho_R\ln\!\Bigl(
     \frac{\mathrm e^{-\beta \hat{H}_{\mathrm{eff}}}}{\operatorname{tr}[\mathrm e^{-\beta \hat{H}_R}]}\Bigr)\Bigr] \\
&=-\operatorname{tr}\bigl[\rho'_R\ln\rho'_R\bigr]
   -\beta\operatorname{tr}\bigl[\hat{H}_{\mathrm{eff}}\rho_R\bigr]
   -\ln\!\bigl(\operatorname{tr}[\mathrm e^{-\beta \hat{H}_R}]\bigr)+\beta\operatorname{tr}\bigl[\hat{H}_{\mathrm{eff}}\rho'_R\bigr]
   -\beta\operatorname{tr}\bigl[\hat{H}_{\mathrm{eff}}\rho'_R\bigr] \\
&=\beta\operatorname{tr}\bigl[\hat{H}_{\mathrm{eff}}(\rho'_R-\rho_R)\bigr]-\operatorname{tr}\bigl[\rho'_R\ln\rho'_R\bigr]
   +\operatorname{tr}\!\Bigl[\rho'_R\ln\!\Bigl(
     \frac{\mathrm e^{-\beta \hat{H}_{\mathrm{eff}}}}{\operatorname{tr}[\mathrm e^{-\beta \hat{H}_R}]}\Bigr)\Bigr] \\
&=\beta\Delta \hat{H}_{\text{eff}}-D(\rho'_R\|\rho_R).
\end{align*}
The mutual information $I(S':R')$ and the relative entropy $D(\rho'_R\|\rho_R)$ are both non‑negative, 
hence $\beta\Delta \hat{H}_{\text{eff}}\geqslant\Delta S$ follows immediately.
\end{proof}

\section{Derivation of Eq. \texorpdfstring{\eqref{ep}}{(18)}}\label{app:entropy}
The expectation value of $\hat{A}$ on the reduced density matrix of the QFT is

\begin{equation}
\begin{aligned}
\text{Tr}(\hat{A}\rho^{(2)}_{f})=\lambda^2\iint d\tau_1 d\tau_2 \left[(1-p)e^{-i\Omega\Delta \tau} + p e^{i\Omega\Delta \tau}\right] &\left[\langle \phi(\tau_1) \hat{A} \phi(\tau_2) \rangle 
- \theta(\Delta\tau) \langle \hat{A} \phi(\tau_1) \phi(\tau_2) \rangle \right. \\
&\left. - \theta(-\Delta\tau) \langle \phi(\tau_1) \phi(\tau_2) \hat{A} \rangle \right].
\end{aligned}
\end{equation}
The effective Hamiltonian can be divided into three parts. First, the number operator $\hat{a}_j^\dagger \hat{a}_j$, which corresponds to the conventional Hamiltonian term. Second, the term $\left[ e^{i\theta_j} (\hat{a}_j^\dagger)^2 
+ e^{-i\theta_j} (\hat{a}_j)^2 \right] $. We only need to compute the expectation value of one of these two terms; the total result can then be expressed as twice the real part of that term. The third part is a constant, whose expectation value clearly vanishes. We first compute the expectation value of $\hat{a}_j^\dagger \hat{a}_j$, obtaining:
\begin{align}
\text{Tr}(\hat{a}^{\dagger}_j \hat{a}_j\rho^{(2)}_{f})&=\lambda^2\left[\left(p\langle \hat{a}_j \hat{a}^{\dagger}_j \rangle- (1-p)\langle \hat{a}^{\dagger}_j \hat{a}_j \rangle\right)|I_{-,j}|^2-\left(p\langle \hat{a}^{\dagger}_j \hat{a}_j \rangle -(1-p)\langle \hat{a}_j \hat{a}^{\dagger}_j \rangle\right)|I_{+,j}|^2\right] \nonumber \\
&+\lambda^2\varepsilon(\Delta \tau)(2-4p)\text{Re}\left( \langle (\hat{a}_j^{\dagger})^2 \rangle I_{+,j}(\tau_1) I_{-,j}(\tau_2)\right)
\label{E}
\end{align}
Here $\varepsilon(\Delta\tau)=1$ for $\Delta\tau>0$ and $\varepsilon(\Delta\tau)=-1$ for $\Delta\tau<0$.
Then, we compute the expectation value of $(\hat{a}_j)^2$, which gives: 
\begin{align}
\text{Tr}((\hat{a}_j)^2\rho^{(2)}_{f})=&\lambda^2(1-p)\Big[  \theta(-\Delta\tau)\left\langle 2\hat{a}_{j}^{2} \right\rangle I_{+,j}^{*}(\tau_{1})I_{+,j}(\tau_{2}) + \theta(\Delta\tau)\left\langle -2\hat{a}_{j}^{2} \right\rangle I_{-,j}(\tau_{1})I_{-,j}^{*}(\tau_{2}) \nonumber\\
           & + \theta(\Delta\tau)\left\langle -2\hat{a}_{j}\hat{a}_{j}^{\dagger}  \right\rangle I_{-,j}(\tau_{1})I_{+,j}(\tau_{2}) + \theta(-\Delta\tau)\left\langle 2\hat{a}_{j}^{\dagger}\hat{a}_{j} \right\rangle I_{-,j}(\tau_{1})I_{+,j}(\tau_{2}) \Big]  \nonumber\\
&+ \lambda^2p\Big[  \theta(-\Delta\tau)\left\langle 2\hat{a}_{j}^{2} \right\rangle I_{-,j}^{*}(\tau_{1})I_{-,j}(\tau_{2}) + \theta(\Delta\tau)\left\langle -2\hat{a}_{j}^{2} \right\rangle I_{+,j}(\tau_{1})I_{+,j}^{*}(\tau_{2})  \nonumber\\
         & + \theta(\Delta\tau)\left\langle -2\hat{a}_{j}\hat{a}_{j}^{\dagger} \right\rangle I_{+,j}(\tau_{1})I_{-,j}(\tau_{2}) + \theta(-\Delta\tau)\left\langle 2\hat{a}_{j}^{\dagger}\hat{a}_{j} \right\rangle I_{+,j}(\tau_{1})I_{-,j}(\tau_{2}) \Big]
\end{align}
\end{widetext}
The real part of the integrand of the function $I_{\pm,j}(\tau_1) I_{\pm,j}^*(\tau_2)$ is an even function with respect to $\Delta \tau$. This implies that, in the context of integration, we have
\begin{equation}
\begin{aligned}
&\mathrm{Re}\left(\theta(\Delta \tau)I_{\pm,j}(\tau_1) I_{\pm,j}^*(\tau_2)\right) \\
&\quad = \mathrm{Re}\left(\theta(-\Delta \tau)I_{\pm,j}(\tau_1) I_{\pm,j}^*(\tau_2)\right) \\
&\quad = \frac{1}{2} |I_{\pm,j}|^2 \,.
\end{aligned}
\end{equation}
Substituting the above results directly into the effective Hamiltonian and using 
\begin{align}
&\langle \left( \hat{a}_{j} \right)^{2}\rangle = \langle ( \hat{a}_{j}^{\dagger} )^{2} \rangle^{*} \nonumber \\
&\quad = -\frac{1}{2} \left( 2\bar{n}_{j} + 1 \right) \sinh\left( 2r_{j} \right) e^{i\theta_{j}}, \\
&\langle \hat{a}_{j}^{\dagger} \hat{a}_{j} \rangle =\langle \hat{a}_{j} \hat{a}_{j}^{\dagger} \rangle - 1 \nonumber \\
&\quad = \bar{n}_{j} \cosh\left( 2r_{j} \right) + \sinh^{2} r_{j},
\end{align}
we have
\begin{widetext}
\begin{equation}
\begin{aligned}
\Delta \hat{H}_{\text{eff}}=\lambda^2\sum_{j} \omega_{j} \Bigg\{ & \left[(1-p)\left(\sinh^{2}r_{j}-\bar{n}_{j}\right)+p\left(\bar{n}_{j}+1+\sinh^{2}r_{j}\right)\right]|I_{-,j}|^{2} \\
& + \left[(1-p)\left(\bar{n}_{j}+1+\sinh^{2}r_{j}\right)+p\left(\sinh^{2}r_{j}-\bar{n}_{j}\right)\right]|I_{+,j}|^{2} \\
& - \sinh(2r_{j}) \operatorname{Re}\left[e^{-i\theta_j} I_{+,j} I_{-,j}\right] \Bigg\}.
\end{aligned}
\end{equation}
Following a similar approach for computing $\Delta S$, we obtain
\begin{align}
    \beta\Delta \hat{H}_{\text{eff}}-\Delta S
    &=\lambda^2\sum_j\left[A_j|I_{-,j}|^2+B_j|I_{+,j}|^2-C_j\operatorname{Re}\left(e^{-i\theta_j}I_{+,j}I_{-,j}\right)\right] \nonumber\\
    &=\lambda^2\sum_j\left[A_j\left|I_{-,j}-\frac{C_j}{2A_j}I_{+,j}^*e^{i\theta_j}\right|^2+\left(B_j-\frac{C_j^2}{4A_j}\right)|I_{+,j}|^2\right],
\label{epb}
\end{align}
where the coefficients are respectively
\begin{align}
    A_{j} &= \sinh^{2} r_{j} \left[ \beta \omega_{j} +(1-2p)(2\bar{n}_{j}+1) \ln \frac{1-p}{p} \right] + [(\bar{n}_j+1)p-\bar{n}_j(1-p)] \left[ \beta \omega_{j} - \ln \frac{1-p}{p} \right] \nonumber\\
    B_{j} &= \sinh^{2} r_{j} \left[ \beta \omega_{j} +(1-2p)(2\bar{n}_{j}+1) \ln \frac{1-p}{p} \right] + [(\bar{n}_j+1)(1-p)-\bar{n}_jp] \left[ \beta \omega_{j} + \ln \frac{1-p}{p} \right] \nonumber\\
    C_{j} &= 2 \sqrt{ \sinh^{2} r_{j}(1 + \sinh^{2} r_{j}) } \left( \beta \omega_{j} +(1-2p)(2\bar{n}_{j}+1) \ln \frac{1-p}{p} \right)
\end{align}
\end{widetext}
Since $0<p<1$, $(1-2p) \ln \frac{1-p}{p}$ is always non-negative. Consequently, both $A_j$ and $B_j$ are increasing functions of $\sinh^{2} r_j$, and their minima are attained at $r_j = 0$. Therefore, we have
\begin{align}
    A_j^{\text{(min)}}&=[(\bar n_j+1)p-\bar n_j(1-p)]\left[\beta\omega_j-\ln\frac{1-p}{p}\right],\nonumber\\
    B_j^{\text{(min)}}&=[(\bar n_j+1)(1-p)-\bar n_jp]\left[\beta\omega_j+\ln\frac{1-p}{p}\right].
\end{align}
By substituting $\beta \omega_{j} = \ln \frac{\bar{n}_j+1}{\bar{n}_j}$, we obtain $A_{j}^{\text{(min)}} \geq 0$ and $B_{j}^{\text{(min)}} \geq 0$. This is the entropy production for the canonical thermal state. Moreover, we can also obtain
\begin{align}
 \beta \omega_{j} +(1-2p)(2\bar{n}_{j}+1) \ln \frac{1-p}{p} = A_{j}^{\text{(min)}}+ B_{j}^{\text{(min)}},
\end{align}
so the coefficients $A_j$, $B_j$, and $C_j$ can be expressed as
\begin{align}
    A_{j} &= \sinh^{2} r_{j} \left[ A_{j}^{\text{(min)}}+B_{j}^{\text{(min)}} \right] + A_{j}^{\text{(min)}}, \nonumber\\
    B_{j} &= \sinh^{2} r_{j} \left[ A_{j}^{\text{(min)}}+B_{j}^{\text{(min)}} \right] + B_{j}^{\text{(min)}}, \nonumber\\
    C_{j} &= 2 \sqrt{ \sinh^{2} r_{j}(1 + \sinh^{2} r_{j}) } \left[ A_{j}^{\text{(min)}}+B_{j}^{\text{(min)}} \right].
\end{align}
We have 
\begin{align}
   4 A_{j} B_{j} -C_{j} ^2= 4A_{j}^{\text{(min)}}B_{j}^{\text{(min)}}\geq 0.
\end{align}
The first line of Eq.~\eqref{epb} is therefore non-negative for all modes. Its completed-square form assumes $A_j>0$; if $A_j=0$, then $C_j=0$ and only $\lambda^2B_j|I_{+,j}|^2\geq0$ remains, as discussed below Eq.~\eqref{ep}.


\begin{thebibliography}{99}
\bibitem{footnote}
For simplicity and to facilitate the analysis of individual momentum modes, we discretize the QFT. This is equivalent to confining the field in a cavity with suitable boundary conditions; the continuum description is recovered in the appropriate limit.
\end{thebibliography}
\end{document}